\documentclass[10pt]{article}

\usepackage{amsthm,amssymb,stmaryrd}
\usepackage{graphics}
\usepackage{amsfonts}
\usepackage{bbm}
\usepackage{tikz}
\usepackage[plainpages=false,pdfpagelabels=false,colorlinks=true,citecolor=blue,hypertexnames=false]{hyperref}
\usepackage{color}

\usepackage[switch]{lineno}

\usepackage{dsfont}
\usepackage{epsfig}
\usepackage{mathrsfs}
\usepackage{amsmath}
\usepackage{cases}

\newtheorem{theorem}{Theorem}
\newtheorem{prop}[theorem]{Proposition}
\newtheorem{thm}[theorem]{Theorem}
\newtheorem{corollary}[theorem]{Corollary}
\newtheorem{lemma}[theorem]{Lemma}

\newtheorem{example}[theorem]{Example}

\makeatletter
\newcommand{\myitem}[1]{%
	\item[(#1)]\protected@edef\@currentlabel{#1}%
}
\makeatother

\def\eatspace#1{#1}
\def\step#1#2{\par\kern1pt\hangindent#2em\hangafter=1\noindent\rlap{\small#1}\kern#2em\relax\eatspace}

\let\set\mathbb
\def\<#1>{\langle#1\rangle}

\def\OO{\mathcal{O}}

\def\val{\operatorname{val}}

\def\sol{\operatorname{Sol}}

\def\diag{\operatorname{diag}}

\def\lc{\operatorname{lc}}
\def\lt{\operatorname{lt}}

\def\a{\alpha}
\begin{document}
\title{Hermite Reduction for D-finite~Functions via Integral Bases
\thanks{S.\ Chen was partially supported by the NSFC
	grants 11871067, 11688101, the Fund of the Youth Innovation Promotion Association, CAS,
	and the National Key Research and Development Project 2020YFA0712300.  L.\ Du
	was supported by the Austrian FWF grant P31571-N32. M. \ Kauers was supported by the Austrian FWF grants P31571-N32 and I6130-N.}}

\author{\bigskip Shaoshi Chen$^{a, b}$, Lixin Du$^{c}$, Manuel Kauers$^c$\\
$^a$KLMM, Academy of Mathematics and Systems Science,\\ Chinese Academy of Sciences,\\ Beijing 100190, China\\
$^b$School of Mathematical Sciences,\\ University of Chinese Academy of Sciences,\\ Beijing 100049, China\\
$^c$Institute for Algebra, Johannes Kepler University,\\ Linz, A4040,  Austria\\
{\sf schen@amss.ac.cn,  lixin.du@jku.at}\\
{\sf manuel.kauers@jku.at}
}

\maketitle

\begin{abstract}
 Trager's Hermite reduction solves the integration problem for algebraic functions via integral bases.
A generalization of this algorithm to D-finite functions has so far been limited to the Fuchsian case.
In the present paper, we remove this restriction and propose a reduction algorithm based on integral
bases that is applicable to arbitrary D-finite functions.
\end{abstract}

\maketitle

\section{Introduction}

Let $R$ be a certain class of functions in one variable $x$ with the derivation $D_x$.
For example, $R$ can be the field of rational functions or algebraic funtions.
In the context of symbolic integration, the \emph{integrability problem} consists in
deciding whether a given element $f\in R$ is of the form $f=D_x(g)$ for some $g\in R$.
If such a $g$ exists, we say that $f$ is \emph{integrable} in~$R$. A relaxed form of
the integrability problem is the \emph{decomposition problem,} which consists in
constructing for a given $f\in R$ elements $g,r\in R$ such that $f=D_x(g)+r$ and $r$
is minimal in a certain sense. Ideally the ``certain sense'' should be such that $r=0$
whenever $f$ is integrable. If $f\in R$ depends on a second variable~$t$, one can also
consider the \emph{creative telescoping} problem: given an element $f\in R$, the task
is to construct $c_0, \ldots, c_r \in R$, not all zero, such that $c_i$ is free of $x$
for all $i \in \{0, \ldots, r\}$ and
\[
c_r D_t^r(f) + \cdots + c_0 f = D_x(g)\quad \text{for some $g\in R$}.
\]
The operator $L = c_r D_t^r + \cdots + c_0$, if it exists, is called a \emph{telescoper}
for~$f$, and $g$ is called a \emph{certificate} for~$L$.

Zeilberger first showed the existence of telescopers for D-finite
functions~\cite{Zeilberger1990}. Almkvist and Zeilberger~\cite{Almkvist1990}
solved the integrability problem and the creative telscoping problem for
hyperexponential functions. Using the adjoint Ore algebra, Abramov and van
Hoeij~\cite{AbramovVanHoeij1999} solved the integrability problem for D-finite
functions. Chyzak~\cite{chyzak00} extended the method of creative telescoping
from hyperexponential functions to general D-finite functions. During the past
ten years, a reduction-based telescoping approach has become popular, which can
find a telescoper without computing the corresponding certificate. This improves
the efficiency of telescoping algorithms because certificates tend to have much
larger size than telescopers and sometimes are not needed. This approach
was first formulated for rational funtions~\cite{bostan10b} and then extended to
hyperexponential functions~\cite{bostan13a}, algebraic functions~\cite{chen16a},
Fuchsian D-finite functions~\cite{chen17a}, and D-finite
functions~\cite{vanderHoeven21,bostan18a}. The reduction-based telescoping
algorithms for algebraic functions and for Fuchsian D-finite functions employ
the notion of integral bases, while the known reduction-based telescoping
algorithms applicable to arbitrary D-finite functions work differently.

The notion of integrality proposed by Kauers and Koutschan~\cite{kauers15b} for Fuchsian
D-finite functions has recently been generalized by Aldossari~\cite{Aldossari20} to
arbitrary D-finite functions, so that the question arises whether there is also
a reduction-based telescoping algorithm for arbitrary D-finite functions based
on integral bases. The purpose of the present paper is to answer this question
affirmatively. This paper is based on the results of Chapter~6 of the second author's Ph.D. thesis~\cite{Du2022}.

First we recall integral bases for D-finite
functions~\cite{kauers15b,Aldossari20} in Section~\ref{sec:integral D-finite}. Then
we extend Hermite reduction for the Fuchsian case to the non-Fuchsian case in
Section~\ref{sec:HR}, which reduces the pole orders of D-finite functions at
finite places. Instead of using polynomial
reduction~\cite{bostan13a,chen16a,chen17a,chen21b} to reduce the pole order at
infinity, we present a Hermite reduction at infinity. Combining Hermite reduction at finite places and at infinity, we are able to determine the integrability of D-finite functions in Section~\ref{sec:add} and present a reduction-based telescoping algorithm for D-finite functions in Section~\ref{sec:telescoper}.

%
\section{Integral Bases}\label{sec:integral D-finite}
Below we recall the value functions and integral bases for arbitrary linear differential operators~\cite{kauers15b, imamoglu17a,Imamoglu-vanHoeij17,Aldossari20}. Let $C$ be a field of characteristic zero and $\bar C$ be the algebraic closure of $C$. Let $C(x)[D]$ be an Ore algebra, where $D$ is the differentiation with respect to $x$ and satisfies the commutation rule $Dx=xD+1$. For an operator $L=\ell_0 + \ell_1 D +\cdots + \ell_n D^n \in C(x)[D]$ with $\ell_n \neq 0$, we consider the left $C(x)[D]$-module $A= C(x)[D]/\<L>$, where $\<L>=C(x)[D]L$. The elements of $A$ are called {\em D-finite} functions.  When there is no ambiguity, an equivalence class $P+\<L>$ in $A$ is also denoted by $P$. Every element of $A$ can be uniquely represented by $P = b_0 + b_1 D +\cdots + b_{n-1} D^{n-1}$ with $b_i\in C(x)$.

For each $\alpha\in  \bar C$, an operator $L$ of order $n$ admits $n$ linearly independent solutions of the form
\begin{equation}\label{EQ:generalized_sol_a}
	(x-\alpha)^\mu \exp(p((x-\alpha)^{-1/s})))b((x-\alpha)^{1/s}, \log(x-\alpha))
\end{equation}
for some $s\in\set N$, $\mu\in \bar C$, $p\in \bar C[x]$ and $b\in\bar  C[[x]][y]$. Such objects are called generalized series solutions at $\alpha$, see~\cite{ince26,KauersPaule2011Book}. For $\alpha=\infty$, the operator $L$ admits $n$ linearly independent solutions of the form
\begin{equation}\label{EQ:generalized_sol_infty}
	x^{-\mu}\exp(p(x^{1/s}))b(x^{-1/s}, \log(x^{-1}))
\end{equation}
for some $s\in\set N$, $\mu\in \bar C$, $p\in \bar C[x]$ and $b\in  \bar C[[x]][y]$. For each $\alpha\in  \bar C\cup \{\infty\}$, let $\sol_\alpha(L)$ be the set of all finite $\bar C$-linear combination of generalized series solutions of $L$ at $\alpha$. Then $\sol_\alpha(L)$ is a $\bar C$-vector space of dimension $n$. Throughout the paper, we assume that for each $\a\in \bar C\cup\{\infty\}$, all series of $\sol_\a(L)$ have $p\in C[x]$, $\mu\in C$ and $b\in C[[x]][y]$ (this can always be achieved by a suitable choice of $C$). If all series of $\sol_\a(L)$ have $p=0$ and $s=1$, then $L$ is called {\em Fuchsian} at $\alpha$. The operator $L$ is simply called {\em Fuchsian} if it is Fuchsian at all $\alpha\in  \bar C\cup\{\infty\}$. In this case, the elements of $A$ are called {\em Fuchsian D-finite} functions.


For simplicity, we assume throughout that $C$ is a subfield of $\set C$. Given two complex numbers $a,b\in \set C$, we say $a\geq b$ if and only if $\text{Re}(a)\geq \text{Re}(b)$. For each $\alpha\in \bar C\cup \{\infty\}$, let $z=x-\alpha$ (or $z=\tfrac{1}{x}$ if $\alpha=\infty$). The {\em valuation} $\nu_\alpha(t)$ of a term $t:=z^r\exp(p(z^{-1/s}))\log(z)^\ell$ is the real part of the local exponent $r$. The {\em valuation} $\nu_\alpha(f)$ of a nonzero generalized series $f$ at $\alpha$ is the minimum of the valuations of all the terms appearing in $f$ (with nonzero coefficients). The {\em valuation} of $0$ is defined as $\infty$. A generalized series $f$ at $\alpha$ is called {\em integral} if $\nu_\alpha(f)\geq0$. A non-integral series $f$ is said to have a {\em pole} at the reference point and its {\em pole order} at $\a$ is defined as $-\nu_\a(f)$. Note that in this terminology, it may happen that $\nu_\a(f') < \nu_\a(f) -1$. For example, $f=\exp(x^{-2})$ is integral at $0$, while the valuation of $f'=-2x^{-3}\exp(x^{-2})$ at $0$ is $-3$, not $-1$. The valuation of a series only depends on its local exponent and not on its exponential part. This valuation is the same as in~\cite[Definition 5.4]{Aldossari20}. 

For each $\alpha\in \bar C\cup\{\infty\}$, an operator $P = b_0 + b_1 D + \cdots +b_{n-1} D^{n-1}$ in $A=C(x)[D]/\<L>$ acts on a generalized series $f\in \sol_\alpha(L)$ via
\[P\cdot f =b_0 f + b_1 f^\prime + \cdots + b_{n-1} f^{(n-1)},\]
where $^\prime$ is the derivation with respect to $x$.
Let $f_1, \ldots, f_n$ be a basis of $\sol_\alpha(L)$ as in the form of~\eqref{EQ:generalized_sol_a} (or~\eqref{EQ:generalized_sol_infty} if $
\a=\infty$). The {\em value function} $\val_\alpha\colon A\to \set R \cup\{\infty\}$ is defined as
\begin{equation}
	\val_\alpha(P) :=\min_{i=1}^n\nu_\a(P\cdot f_i).
\end{equation}
Then $\val_\alpha(P)$ is the minimum valuation of all series $P\cdot f$ at $\alpha$, where $f$ runs through all series solutions in $\sol_\alpha(L)$. So this definition of value functions is independent of the choice of the basis of $\sol_\alpha(L)$. 
An element $P\in A$ is called {\em (locally) integral} at $\alpha\in  \bar C\cup \{\infty\}$ if $\val_\alpha(P)\geq 0$. If $P$ is not locally integral at $\a$, then $P$ is said to have a pole at $\a$ and its {\em pole order} at $\a$ is defined as $-\val_\a(P)$. An element $P\in A$ is called {\em (globally) integral} if $\val_\alpha(P)\geq 0$ for all $\alpha\in \bar C$, i.e., $P$ is locally integral at all finite places. When $L$ is Fuchsian, this notion of integrality falls back to the Fuchsian case discussed in~\cite{chen17a,kauers15b}.

The set of all globally integral elements $f\in A= C(x)[D]/ \<L>$ forms a $C[x]$-module. A basis of this module is called a {\em (global) integral basis} for $A$. Such bases exist and the algorithm for computing integral bases in the Fuchsian case~\cite{kauers15b} applies to the setting of the non-Fuchsian case literally. More properties can be found in~\cite{Aldossari20}. 

For a fixed $\alpha\in \bar C$, the {\em valuation} $\nu_\alpha$ of a nonzero rational function $f\in C(x)$ is an integer $m\in \set Z$ such that $f = (x-\alpha)^m p/q$ with $p,q\in C[x]$, $\gcd(p,q)=1$ and $(x-\a) \nmid pq$. By convention, set $\nu_\alpha(0)=\infty$. The {\em valuation} $\nu_\infty$ of a rational function $f=p/q\in C(x)$ is $\deg_x(q)-\deg_x(p)$. For each $\alpha\in\bar C\cup \{\infty\}$, the valuation $\nu_\alpha$ of a rational function is the same as the valuation of its Laurent series expansion at $\alpha$. The set $ C(x)_\alpha=\{f\in C(x) \mid \nu_\alpha(f)\geq 0\}$ forms a subring of $C(x)$. The set of all elements $f\in A$ that are locally integral at some fixed $\alpha\in \bar C\cup \{\infty\}$ forms a $ C(x)_\alpha$-module. A basis of this module is called a {\em local integral basis} at $\alpha$ of $A$. Such a basis can also be computed~\cite{kauers15b,Aldossari20}.


%




An integral basis $\{\omega_1,\ldots,\omega_n\}$ is always a $C(x)$-vector space basis of $A$. A key feature of integral bases is that they make poles explicit. Writing an element $f\in A$ as a  combination  $f= \sum_{i=1}^nf_i\omega_i$ for some $f_i\in C(x)$, we have that $f$ has a pole at $\alpha\in  \bar C$ if and only if at least one of the $f_i$ has a pole there. Furthermore, $\lfloor\val_\alpha(f)\rfloor$ is a lower bound for the valuations of all the $f_i$'s at $\alpha$.
\begin{lemma}\label{Lem:pole}
	Let $\{\omega_1, \ldots, \omega_n\}$ be a local integral basis of $A$ at some fixed $\alpha\in  \bar C \cup\{\infty\}$.  Let $f\in A$
	and $f_1,\ldots, f_n\in C(x)$ be such that $f=\sum_{i=1}^nf_i\omega_i$. Then
	\begin{enumerate}
		\item $f$ is integral at $\alpha$ if and only if for each $i\in\{1,\ldots, n\}$, $f_i\omega_i$ is integral at $\a$.
		
		\item $\lfloor\val_\alpha(f)\rfloor = \min_{i=1}^n \nu_\alpha(f_i)$.
	\end{enumerate}
\end{lemma}
\begin{proof}
	$(1)$: The direction ``$\Leftarrow$'' is obvious. To show ``$\Rightarrow$'', suppose
	that $f$ is integral at~$\alpha$. Then there exist $\tilde{f}_1,\dots,\tilde{f}_n\in C(x)_\alpha$ such that
	$f=\sum_{i=1}^n\tilde{f}_i\omega_i$. Thus $\sum_{i=1}^n(\tilde{f}_i-f_i)\omega_i=0$, and then
	$\tilde{f}_i=f_i$ for all~$i$, because $\{\omega_1,\dots,\omega_n\}$ is a $C(x)$-vector space of~$A$.
	As elements of $C(x)_\alpha$, the $f_i$'s are integral at~$\alpha$. Hence the $f_i\omega_i$'s
	are integral at~$\alpha$.
	
	$(2)$: Let $\tau :=\min_{i=1}^n \nu_\alpha(f_i)$. We have to show that $\tau$ is an integer such that
	\[\tau \leq \val_\alpha(f)< \tau+1.\]
	Let $z\in \bar C(x)$ with $\nu_\a(z)=1$. Since $z^{-\tau}f_i\omega_i$ is integral at $\a$, we have $z^{-\tau}f$ is integral at $\alpha$. Thus $\val_\alpha (z^{-\tau}f) =\val_\alpha(f)-\tau \geq 0$, which implies $\tau \leq\val_\alpha(f)$. On the other hand, suppose $\val_\alpha(f)\geq \tau+1$. Then $z^{-(\tau+1)}f$ is integral at $\alpha$. However, $z^{-(\tau+1)}f$ does not belong to the $ C(x)_\alpha$-module generated by $
	\{\omega_1, \ldots, \omega_n\}$ because there is $i\in\{1,\ldots, n\}$ such that $\tau=\nu_\alpha(f_i)$ and then $z^{-(\tau+1)}f_i \notin  C(x)_\alpha$. This contradicts the fact that $\{\omega_1,\ldots,\omega_n\}$ is a local integral basis at $\alpha$.
\end{proof}
Let $W=(\omega_1,\dots,\omega_n)$ be a vector space basis of~$A$ over $C(x)$. For $f\in A$, denote its derivative $D\cdot f$ by $f'$. Let $e\in C[x]$ be and $M=(m_{i,j})_{i,j=1}^n\in C[x]^{n\times n}$ be such that $eW'=MW$ and $\gcd(e, m_{1,1}, m_{1,2},\ldots,m_{n,n})=1$. If $W$ is an integral basis and $L$ is Fuchsian at all finite places, then $e$ must be squarefree, see~\cite[Lemma 3]{chen17a}. If $W$ is a local integral basis at infinity and $L$ is Fuchsian at infinity, then $\deg_x(m_{i,j})<\deg_x(e)$ for all $i,j$, see~\cite[Lemma 4]{chen17a}. However, these two facts are no longer true in the non-Fuchsian case, as the following examples show:
\begin{example}\label{Ex:non-fuch at 0}
	The operator $L = x^3 D^2 + (3x^2+2) D\in\set C(x)[D]$ has only one singular point $0$ in $\set C$, which is an irregular singular point. Infinity is a regular singular point. So $L$ is Fuchsian at all points in $\set C\cup\{\infty\}$ except $0$ and this implies that $L$ is a non-Fuchsian operator. At the point $0$, there are two linearly independent solutions $y_1(x) = 1$ and $y_2(x) = \exp(x^{-2})$ in $\sol_0(L)$. An integral basis for $A=\set C(x)[D]/\<L>$ is given by $\omega_1 = 1$ and $\omega_2 = x^3 D$, which is also a local integral basis at infinity. Then
	\begin{equation*}
		\begin{pmatrix}
			\omega_1'\\
			\omega_2'
		\end{pmatrix} = \frac{1}{e}
		\begin{pmatrix}
			0&1\\
			0&-2\\
		\end{pmatrix}
		\begin{pmatrix}
			\omega_1\\
			\omega_2
		\end{pmatrix}
	\end{equation*}
	with $e=x^3$. In this example, $e$ is not squarefree.
\end{example}
\begin{example}\label{Ex:non-fuch at infty}
	Let  $L = x D^2 - (3x^3+2) D\in\set C(x)[D]$. Infinity is an irregular singular point. So $L$ is not Fuchsian at infinity. There are two linearly independent solutions $y_1(x) = 1$ and $y_2(x) = \exp(x^3)$ in $\sol_\infty(L)$. A local integral basis at infinity of $A=\set C(x)[D]/\<L>$ is given by $\omega_1 = 1$ and $\omega_2 = x^{-2}D$. Then
	\begin{equation*}
		\begin{pmatrix}
			\omega_1'\\
			\omega_2'
		\end{pmatrix} =
		\begin{pmatrix}
			0&x^2\\
			0&3x^2\\
		\end{pmatrix}
		\begin{pmatrix}
			\omega_1\\
			\omega_2
		\end{pmatrix}.
	\end{equation*}
	In this example, $e=1$ and the condition $\deg_x(m_{i,j}) < \deg_x(e)$ fails.
	
\end{example}

A $C(x)$-vector space basis $\{\omega_1,\dots,\omega_n\}$ of $A=C(x)[D]/\<L>$ is
called \emph{normal} at $\a\in \bar C\cup\{\infty\}$ if there exist $r_1,\dots,r_n\in
C(x)$ such that $\{r_1\omega_1,\dots,r_n\omega_n\}$ is a local integral basis at~$\a$. Given an integral basis and a local integral basis at infinity, Trager~\cite{Trager84} presented an algorithm for computing an integral basis that is normal at infinity in the algebraic function field. The same procedure also applies in the present situation, see~\cite[Algorithm 5.20]{Aldossari20}. 

\section{Hermite Reduction}\label{sec:HR}

Hermite reduction, first introduced by Ostrogradsky in 1845~\cite{Ostrogradsky1845}, is a classical symbolic integration technique that reduces rational functions to integrands with only simple poles. Hermite reduction was extended by Trager~\cite{Trager84} from the field of rational functions to that of algebraic functions via integral bases. Trager's Hermite reduction solved the integration problem for algebraic functions. This work
was extended to the case of Fuchsian D-finite functions~\cite{chen17a}.
We shall further extend Hermite reduction to general D-finite functions, including the non-Fuchsian case. To reduce the pole order at infinity, we develop a Hermite reduction at infinity for D-finite functions, which plays the same role as polynomial reduction~\cite{bostan13a,chen16a,chen17a,chen21b}. 
In this section, Hermite reduction at finite places and at infinity are formulated in the same framework. More precisely, we shall use a local integral basis at $\alpha\in \bar C\cup\{\infty\}$ to reduce the pole orders of D-finite functions at~$\alpha$.

For convenience, we introduce some notations for the valuations of a matrix with rational coefficients. For each $\alpha\in \bar C\cup\{\infty\}$, the {\em valuation} of a matrix $T\in C(x)^{n\times n}$ at~$\alpha$, denoted by $\nu_\alpha(T)$, is defined as the minimal valuation at $\alpha$ of all entries in this matrix. The {\em degree} of $T\in C(x)^{n\times n}$, denoted by $\deg_x(T)$, is defined as $-\nu_\infty(T)$. In particular, the degree of a rational function $f=p/q\in C(x)$ is $\deg_x(p)-\deg_x(q)$.

\subsection{The Local Case}\label{sec:local}
Let $L\in C(x)[D]$ be of order $n$ and let $A= C(x)[D]/\<L>$. For an arbitrary but fixed point $\alpha\in \bar C\cup\{\infty\}$, let $W=(\omega_1,\dots,\omega_n)$
be a local integral basis at $\alpha$ of~$A$ and then there exists a matrix $T\in C(x)^{n\times n}$ such that $W'=TW$. We write $z=x-\alpha$ (or $z=\tfrac{1}{x}$ if $\alpha=\infty$).
Let $\lambda= -\nu_\alpha(T)$ be the pole order of $T$ at $\alpha$. Then $\lambda\in \set Z$ and there exists a matrix $M=(m_{i,j})_{i,j=1}^n\in \bar C(x)_\alpha^{n\times n}$ such that
\[W'=\frac{1}{z^\lambda} MW\quad\text{and} \quad\nu_{\alpha}(M)=0,\]
where $M=z^{\lambda}T$.
Let $f=\frac{1}{z^d}\sum_{i=1}^n a_i\omega_i\in A$ with
$d>1$ (or $d\geq 0$ if $\a=\infty$) and $a_1,\dots,a_n\in \bar C(x)_\alpha$. 
In order to reduce the pole order $d$ of $f$ at $\alpha$, we seek
$b_1,\dots,b_n,c_1,\dots,c_n\in \bar C(x)_\alpha$ such that
\begin{equation}\label{EQ:hr-goal}
	\frac{1}{z^d}\sum_{i=1}^na_i\omega_i =\left(\frac{1}{z^{d+\mu}}\sum_{i=1}^nb_i\omega_i\right)'+ \frac{1}{z^{d-1}}\sum_{i=1}^nc_i\omega_i,
\end{equation}
where $\mu\in \set Z$ is an integer such that $\nu_\alpha(z') =\nu_\alpha(z)+\mu$. In this setting, $\mu = -1$ if $\a\in \bar C$ (because $(x-\a)'=1$); $\mu=1$ if $\a=\infty$ (because $(\tfrac{1}{x})'=-\frac{1}{x^2}$). Also $z'=-\mu z^{\mu+1}$.

After differentiating both sides of~\eqref{EQ:hr-goal} and multiplying by $z^d$, we get
\begin{align}\label{EQ:hr-expand1}
	\sum_{i=1}^n a_i \omega_i
	&=\sum_{i=1}^n \left(\frac{b_i^\prime}{ z^\mu}\omega_i + b_iz^{d}\left(\frac{\omega_i}{z^{d+\mu}}\right)^\prime + c_iz \omega_i\right)\\\label{EQ:hr-expand2}
	&=\sum_{i=1}^n \left(\frac{b_i^\prime}{ z^\mu}\omega_i + \frac{b_i}{z^{\lambda+\mu}} \sum_{j=1}^nm_{i,j} \omega_j +\mu (d+\mu)b_i\omega_i + c_iz \omega_i\right),
\end{align}
where $\mu (d+\mu) = z^d (z^{-(d+\mu)})'$. Note that $b_i$ is integral at $\alpha$. Then $b_i'z^{-\mu}\in z \bar C(x)_\alpha$ because $\nu_\a(b_i'z^{-\mu})\geq 1$. For example, if $\a\in \bar C$, then $(1+ (x-\a)+\cdots)'(x-\a)=(x-\a)+\cdots$; if $\a=\infty$, then $(1+\frac{1}{x}+\cdots)'x=-\frac{1}{x}+\cdots$. So if $-(\lambda +\mu) >0$, i.e., $\lambda <-\mu $, then Equation~\eqref{EQ:hr-expand2} can be reduced modulo $z$:
\begin{equation}\label{EQ:hr-simple-case}
	\sum_{i=1}^n a_i \omega_i \equiv \sum_{i=1}^n \mu(d+\mu)b_i\omega_i\mod z.
\end{equation}
It follows that $b_i \equiv \mu(d+\mu)^{-1} a_i\mod z$ is the unique solution of~\eqref{EQ:hr-simple-case} in $\bar C(x)_\a/\<z>$.
If $\lambda\geq -\mu$, then multiplying~\eqref{EQ:hr-expand1} by $z^{\lambda+\mu}$ and reducing this equation modulo $z^{\lambda+\mu+1}$ yields
\begin{equation}\label{EQ:hr-mod}
	\sum_{i=1}^n z^{\lambda+\mu} a_i\omega_i \equiv \sum_{i=1}^n b_iz^{d+\lambda+\mu}\left(\frac{\omega_i}{z^{d+\mu}}\right)^\prime \mod z^{\lambda+\mu+1}.
\end{equation}
Let $\psi_i := z^{d+\lambda+\mu}\left(\frac{\omega_i}{z^{d+\mu}}\right)^\prime$ for $i=1,\ldots, n$. To perform Hermite reduction, we have to show that~\eqref{EQ:hr-mod}
always has a solution $(b_1,\ldots, b_n)$ in $\left(\bar C(x)_\a/\<z^{\lambda+\mu+1}>\right)^n$.

In the Fuchsian case, Chen et al.~\cite{chen17a} proved that $\lambda =1$. When $\a\in \bar C$, they showed that $\{\psi_1, \ldots, \psi_n\}$ forms a local integral basis at $\a$ and hence~\eqref{EQ:hr-mod} has a solution and that solution is unique. When $\a=\infty$, instead of solving the modular system~\eqref{EQ:hr-mod}, they introduced the polynomial reduction to reduce the degree in~$x$. We shall show that $\{\psi_1, \ldots, \psi_n\}$ still forms a local integral basis at infinity. Then the polynomial reduction can be formulated as Hermite reduction at infinity, as suspected by one of the anonymous referees of~\cite{chen17a}.

In the non-Fuchsian case, it may happen that $\lambda>1$, see Examples \ref{Ex:non-fuch at 0} (for $\a=0$) and~\ref{Ex:non-fuch at infty} (for $\a=\infty$). Another difference is that $\{\psi_1, \ldots, \psi_n\}$ may not be a local integral basis at $\a$ anymore, see the following Example~\ref{Ex:psi}. Fortunately, the linear system~\eqref{EQ:hr-mod} still has a unique solution in $\left(\bar C(x)_\a/\<z^{\lambda+\mu+1}>\right)^n$ as we shall prove in this section. There are two steps. First we show that $\{\psi_1, \ldots, \psi_n\}$ is linearly independent over $\bar C(x)$ and then we find a rational solution $(b_1,\ldots, b_n)$ whose entries admit nonnegative valuation at $\a$. So the $b_i$'s belong to $\bar C(x)_\a$. Taking $b_i$ modulo $z^{\lambda +\mu +1}$ gives a unique solution of~\eqref{EQ:hr-mod}. Once we know that~\eqref{EQ:hr-mod} has a solution, equating the coefficients of the $\omega_i$'s on both sides, we can find its solution $b=(b_1, \ldots, b_n)$ by solving the following linear system of congruence equations:
\begin{equation}\label{EQ:hr-matrix}
	(z^{\lambda+\mu} a_1,\ldots, z^{\lambda +\mu }a_n) \equiv b(M +\mu(d+\mu)z^{\lambda+\mu}I_n) \mod z^{\lambda+\mu+1},
\end{equation}
where $I_n$ is the identity matrix in $C[x]^{n\times n}$.

\begin{example}\label{Ex:psi}
	Continue Example~\ref{Ex:non-fuch at 0}. For $\a=0$ and $\lambda=3$, let $\psi_i = x^{d+2}(x^{1-d}\omega_i)'$ for $i=1,2$. A direct calculation yields that
	\begin{equation}\label{EQ:psi}
		\begin{pmatrix}
			\psi_1\\
			\psi_2
		\end{pmatrix}
		= \begin{pmatrix}
			-(d-1)x^2 &1\\
			0&-(d-1)x^2-2
		\end{pmatrix}
		\begin{pmatrix}
			\omega_1\\
			\omega_2
		\end{pmatrix}.
	\end{equation}
	In this example, $\psi_1,\psi_2$ are integral elements but do not form a local integral basis at $0$, because $\frac{1}{x^2}(2\psi_1+\psi_2)=-2(d-1)\omega_1-(d-1) \omega_2$ is integral at $0$. In fact, if $d>1$, then $\{\psi_1,\frac{1}{x^2}(2\psi_1+\psi_2)\}$ is a local integral basis at $0$. Now~\eqref{EQ:hr-matrix} becomes
	\[(a_1x^2, a_2x^2)\equiv (b_1, b_2)\begin{pmatrix}
		-(d-1)x^2&1\\
		0&-(d-1)x^2-2\\
	\end{pmatrix}\mod x^3.\]
	When $d>1$, even though the coefficient matrix is not invertible over $\bar C(x)_0/\<x^3>$, this equation still has a unique solution
	\[\left\{\begin{aligned}
		&b_1\equiv -(d-1)^{-1} a_1\mod x^3,\\
		&b_2\equiv \frac{1}{4}\left((d-1)x^2-2\right)\left(a_2x^2 +(d-1)^{-1}a_1\right) \mod x^3.
	\end{aligned}\right.\]
\end{example}

Let $\OO_\alpha$ denote the set of all elements in $\bar C(x)[D]/\<L> $ that are locally integral at $\alpha\in  \bar C\cup\{\infty\}$. Even though $\{\psi_1, \ldots, \psi_n\}$ may not be a local integral basis at $\a$, it is not so far away. In Example~\ref{Ex:psi}, we have $\bar C(x)_0 \psi_1 + \bar C(x)_0 \psi_2 \subseteq \OO_0 \subseteq \frac{1}{x^{2}} (\bar C(x)_0 \psi_1 +\bar  C(x)_0 \psi_2)$. In general, if we represent a locally integral element at $\a$ as a linear combination of $\{\psi_1,\ldots, \psi_n\}$ with coefficients in $\bar C(x)$, the pole orders at $\a$ of these coefficients are at most $\lambda+\mu$.

\begin{prop}\label{Prop:psi}
	Let $\a\in \bar C\cup \{\infty\}$ and $W=\{\omega_1, \ldots, \omega_n\}$ be a local integral basis at $\a$ of $A$. Let $z= x-\a$ (or $z=x^{-1}$ if $\a = \infty$) and $\mu\in \set Z$ be such that $\nu_\a(z') = \nu_\a(z)+\mu$. Let $\lambda\in \set Z$ and $M\in  \bar C(x)_\a^{n\times n}$ be such that $z^\lambda W^\prime =MW$ and $\nu_\a(M)=0$. For some integer $d>1$ (or $d\geq 0$ if $\a=\infty$), we define $\psi_i := z^{d+\lambda+\mu}(z^{-d-\mu}\omega_i)'$. If $\lambda \geq -\mu$, then
	\[\sum_{i=1}^n \bar C(x)_\a \psi_i \subseteq \OO_\a \subseteq \frac{1}{z^{\lambda+\mu}}\sum_{i=1}^n \bar  C(x)_\a \psi_i.\]
	In particular, when $\lambda=-\mu$, we have $\sum_{i=1}^n \bar C(x)_\a\psi_i =\OO_\a$. In this case, $\{\psi_1,\ldots, \psi_n\}$ forms a local integral basis at $\a$.
\end{prop}
\begin{proof}
	We prove this proposition using the same technique as in~\cite[Lemma 10]{chen16a}. To show $\sum_{i=1}^n \bar C(x)_\a \psi_i \subseteq \OO_\alpha$, we only need to show that for every $i=1,\ldots,n$, the element $\psi_i$ is integral at $\a$. After differentiating, we get $\psi_i = z^\lambda \omega_i^\prime +\mu(d+\mu)z^{\lambda+\mu}\omega_i$. Since $z^\lambda W^\prime = MW$ and $\nu_\a(M)=0$, it follows that $z^\lambda\omega_i'$ is integral at $\alpha$. Then $\psi_i$ is integral at $\a$ because $\lambda+\mu\geq 0$.
	
	Next we shall prove $\OO_\alpha \subseteq \frac{1}{z^{\lambda+\mu}}\sum_{i=1}^n \bar C(x)_\alpha \psi_i$. Suppose to the contrary that there exists an element $f\in \OO_\alpha\setminus\frac{1}{z^{\lambda+\mu}}\sum_{i=1}^n \bar C(x)_\a\psi_i$. Furthermore, we can find such an element $f$ of the form
	\[f= \frac{1}{z^{\lambda+\mu+1}}\sum_{i=1}^n c_i \psi_i\quad\text{with } c_i\in \bar C(x)_\a \text{ and } \nu_\a(c_i)=0 \text{ for some } i.\]
	Let $g=z^{-\mu-1}\sum_{i=1}^n c_i^\prime \omega_i$, which is integral. Then also their sum
	\begin{align*}
		f+g &= z^{d-1}\sum_{i=1}^n \left(c_i\bigl(z^{-d-\mu}\omega_i\bigr)'
		+ c_i'z^{-d-\mu}\omega_i \right) \\
		&= z^{d-1} \sum_{i=1}^n \bigl(c_iz^{-d-\mu}\omega_i\bigr)'=z^{d-1}\left(z^{-d-\mu}h\right)'
	\end{align*}
	must be integral, where $h=\sum_{i=1}^nc_i\omega_i$. Since $\{\omega_1,\ldots,\omega_n\}$ is an integral basis at~$\alpha$, by Lemma~\ref{Lem:pole} we have $0 \leq \val_\a(h) <1$. There exists a generalized series solution $y_i\in\sol_\a(L)$ such that $h\cdot y_i$ involves a term
	$T=z^r \exp(p(z^{-1/s}))\log(z)^\ell$ with $0\leq r<1$, $s,\ell\in\set N$ and $p\in \bar C[x]$. For this fixed series $y_i$, let $T$ be the dominant term of $h\cdot y_i$,
	i.e., among all terms with minimal $r$ the one with the largest exponent~$\ell$. Let $k = \deg_x(p)$ and $c=\lc_x(p)$ be the degree and the leading coefficient of $p$ in $x$ respectively. Then
	\begin{align}\nonumber
		&\bigl(z^{d-1} D z^{-d-\mu}\bigr) \cdot T
		\\\nonumber
		={}&	\mu(d-r+\mu)z^{r-1}\exp(p(z^{-1/s}))\log(z)^\ell \\\nonumber
		&+\tfrac{\mu c k}{s}z^{r-\frac{k}{s}-1}\exp(p(z^{-1/s}))\log(z)^\ell +\cdots\\\label{eq:T}
		&-\mu\ell z^{r-1}\exp(p(z^{-1/s}))\log(z)^{\ell-1},
	\end{align}
	where ``$\cdots$'' denotes some terms of valuation higher than $r-\frac{k}{s}-1$.
	
	If $k=0$, then \[\mu(d-r+\mu)z^{r-1}\exp(p(z^{-1/s}))\log(z)^\ell\] is the dominant term of
	$\bigl(z^{d-1} D z^{-d-\mu}\bigr)\cdot (h\cdot y_i)$; here we use the
	assumption that $d>1$ (resp. $d\geq 0$ if $\a=\infty$), because for $d=1$ (resp. $d=-1$) and $r=0$ the
	coefficient $(d-r+\mu)$ is zero.
	
	If $k>0$, then
	\[\frac{\mu ck}{s}z^{r-\frac{k}{s}-1}\exp(p(z^{-1/s}))\log(z)^\ell\]
	is the dominant term of
	$\bigl(z^{d-1} D z^{-d+\mu}\bigr)\cdot (h\cdot y_i)$.
	
	The above calculation reveals that
	$z^{d-1}\bigl(z^{-d-\mu}h\bigr)'=f+g$ is not integral at~$\alpha$, which contradicts our assumption on
	the integrality of~$f$. Hence $\OO_\alpha \subseteq \frac{1}{z^{\lambda+\mu}}\sum_{i=1}^n \bar C(x)_\alpha \psi_i$.
\end{proof}

\begin{thm}\label{Thm:hr}
	Use the same notations as in Proposition~\ref{Prop:psi}. Let $d>1$ (or $d\geq0$ if $\alpha=\infty$). If $\lambda \geq -\mu$, then for any $a_1,\ldots, a_n\in \bar C(x)_\a$, the linear system
	\begin{equation}\label{EQ:linear-sys}
		\sum_{i=1}^n z^{\lambda+\mu}a_i \omega_i = \sum_{i=1}^n b_i\psi_i
	\end{equation}
	has a unique solution $(b_1, \ldots, b_n)$ in $\left(\bar C(x)_\a/\<z^{\lambda+\mu+1}>\right)^n$.
\end{thm}
\begin{proof}
	By Proposition~\ref{Prop:psi}, the $\bar C(x)_\a$-module generated by \[\left\{\frac{1}{z^{\lambda+\mu}}\psi_1, \ldots, \frac{1}{z^{\lambda+\mu}}\psi_n\right\}\]
	contains a submodule $\OO_\a$ of rank $n$. So $\{\psi_1, \ldots, \psi_n\}$ is linearly independent over $\bar C(x)$. Then there exist unique $t_1,\ldots, t_n\in \bar C(x)$ such that
	$\sum_{i=1}^n z^{\lambda+\mu} a_i \omega_i =\sum_{i=1}^n t_i\psi_i$. 
	
	To find a solution $b_i$, we have to show that $t_i\in \bar C(x)_\a$ for all $i=1,\ldots, n$. If so, $b_i\equiv t_i\mod z^{\lambda+\mu+1}$ is the unique solution of~\eqref{EQ:linear-sys}. Since $a_i\in\bar C(x)_\a$ and the $\omega_i$'s are integral at $\a$, the element
	\[\sum_{i=1}^na_i\omega_i =\frac{1}{z^{\lambda+\mu}}\sum_{i=1}^nt_i\psi_i\]
	is integral at $\a$. By Proposition~\ref{Prop:psi},
	\[\frac{1}{z^{\lambda+\mu}}\sum_{i=1}^nt_i\psi_i \in \OO_\a \subseteq \frac{1}{z^{\lambda+\mu}}\sum_{i=1}^n\bar C(x)_\a \psi_i.\]
	Then $\sum_{i=1}^nt_i\psi_i \in \sum_{i=1}^n\bar C(x)_\a \psi_i.$ Since $\{\psi_1, \ldots, \psi_n\}$ are linearly independent over $\bar C(x)$, we have $t_i\in  \bar C(x)_\a$ for all $i$. Thus $t_i\in \bar C(x)_\a$ as claimed.
\end{proof}

According to Theorem~\ref{Thm:hr}, we can perform one step of Hermite reduction for D-finite functions as described in the beginning of this section. The element $b_i$ in $\bar C(x)_\a/\<z^{\lambda+\mu+1}>$ is of the form
\[b_i = b_{i,0} + b_{i,1} z + \cdots + b_{i, \lambda + \mu} z^{\lambda + \mu} \quad \text{with}\quad b_{i,j}\in \bar C.\]
So in Equation~\eqref{EQ:hr-goal}, for $\alpha\in \bar C$, if $d > \max\{1, \lambda\}$, then we can guarantee that the coefficients of $\frac{1}{z^{d+\mu}}\sum_{i=1}^nb_i\omega_i$ are proper rational functions. For $\alpha=\infty$, if $d \geq \max\{0, \lambda\}$, then the coefficients of $\frac{1}{z^{d+\mu}}\sum_{i=1}^nb_i\omega_i$ are polynomials. 

\begin{example}\label{Ex:hr-finite}
	Continue Examples \ref{Ex:non-fuch at 0} and~\ref{Ex:psi}. A local integral basis at $\a=0$ is given by $\omega_1=1$ and $\omega_2=x^3D$. Then $\lambda =3$. Consider the D-finite function
	\[f = \frac{(-2x^2-x^4)\omega_1 + (-2+3x^2-3x^4)\omega_2}{x^4}\]
	and use Hermite reduction at $0$ to reduce the power of $x$ in its denominator. So we start with $z=x$, $\mu=-1$, $d=4$, $a_1 = -2x^2-x^4$ and $a_2=-2+3x^2-3x^4$. From~\eqref{EQ:hr-matrix}, we get
	\[(a_1x^2, a_2x^2)\equiv (b_1, b_2)\begin{pmatrix}
		-3x^2&1\\
		0&-3x^2-2\\
	\end{pmatrix}\mod x^3.\]
	By Theorem~\ref{Thm:hr}, we know that this equation has a unique solution. Indeed, we find a solution $b_1= \tfrac{2}{3}x^2$, $b_2=\tfrac{4}{3}x^2$. Then one step of the Hermite reduction at $0$ simplifies $f$ to
	\[f=\left(\frac{2\omega_1+4\omega_2}{3x}\right)'+\frac{(-4-3x^2)\omega_1+(13-9x^2)\omega_2}{3x^2}.\]
\end{example}

\begin{example}\label{Ex:hr-infty}
	Let  $L = x D^2 - (3x^3+2) D\in\set C(x)[D]$ be the same operator as in Example~\ref{Ex:non-fuch at infty}. A local integral basis at $\a=\infty$ of $A=\set C(x)[D]/\<L>$ is given by $\omega_1 = 1$ and $\omega_2 = x^{-2}D$. Then
	\begin{equation*}
		\begin{pmatrix}
			\omega_1'\\
			\omega_2'
		\end{pmatrix} =
		x^\lambda\begin{pmatrix}
			0&1\\
			0&3\\
		\end{pmatrix}
		\begin{pmatrix}
			\omega_1\\
			\omega_2
		\end{pmatrix}
	\end{equation*}
	with $\lambda=2$. Consider the D-finite function
	\[f = 4x^3 + \frac{1}{x} D =4x^3\omega_1 + x\omega_2=x^3\left(4\omega_1 + \frac{1}{x^2}\omega_2\right)\]
	and use Hermite reduction at infinity to reduce its degree in $x$. So we start with $z=\tfrac{1}{x}$, $\mu=1$, $d=3$, $a_1 = 4$ and $a_2=\tfrac{1}{x^2}$. From~\eqref{EQ:hr-matrix}, we get 
	\[(a_1x^{-3}, a_2x^{-3})\equiv (b_1, b_2)\begin{pmatrix}
		4x^{-3}&1\\
		0&4x^{-3}+3\\
	\end{pmatrix}\mod x^{-4}\]
	This coefficient matrix is not invertible over $C(x)_{\infty}/\<x^{-4}>$. However, by Theorem~\ref{Thm:hr}, we know this equation has a unique solution. Indeed, we find a solution $b_1= 1$, $b_2=\tfrac{4}{9x^3}-\tfrac{1}{3}$. Then one step of the Hermite reduction at infinity simplifies $f$ to
	\[f=\left(x^4\omega_1+\left(\frac{4}{9}x-\frac{1}{3}x^4\right)\omega_2\right)'+\left(x -\frac{4}{9}\right)\omega_2.\]	
	
\end{example}

For a rational function $g\in C(x)$ and any point $\a\in \bar C$, if $\nu_\a(g)\neq 0$, we have $\nu_\a(g') = \nu_\a(g)-1$.
So the pole order of a rational function increases by exactly one under each derivation. In the D-finite case, the pole order increases by at least one. A lower bound of its pole order under each derivation is given in~\cite[Lemma 5.7]{Aldossari20}.
\begin{lemma}\label{Lem:val_under_der}
	Let $g\in A$. For any $\a\in  \bar C\cup\{\infty\}$, if $\val_\a(g)\neq 0$, then $\val_\a(g') \leq \val_\a(g)+\mu$, where $\mu=-1$ if $\alpha\in \bar C$ and $\mu = 1$ if $\alpha = \infty$.
\end{lemma}
\begin{proof}
	Let $y_i$ be a generalized series solution in $\sol_\a(L)$ such that $\val_\a(g)=\nu_\a(g\cdot y_i)$. Let $z=x-\a$ (or $z=1/x$ if $\a=\infty$). Let $T=z^r \exp(p(z^{-1/s}))\log(z)^\ell$ with $r\neq 0$, $s,\ell\in\set N$, $p\in \bar C[x]$ be the dominant term of $f\cdot y_i$,
	i.e., among all terms with minimal $r$ the one with the largest exponent~$\ell$. Let $k = \deg_x(p)$ and $c=\lc_x(p)$. Then
	\begin{align}\nonumber
		D \cdot T
		& =
		-r\mu z^{r+\mu}\exp(p(z^{-1/s}))\log(z)^\ell \\\nonumber
		&+\frac{\mu ck}{s}z^{r-\frac{k}{s}+\mu}\exp(p(z^{-1/s}))\log(z)^\ell + \cdots\\\label{eq:DT}
		&-\ell\mu z^{r+\mu}\exp(p(z^{-1/s}))\log(z)^{\ell-1}.
	\end{align}
	Note that $r-\frac{k}{s}+\mu \leq r+\mu$ and $r\neq 0$ (by the assumption $\val_\a(g)\neq 0$). So the valuation of the term $D\cdot T$ in $g'\cdot y_i$ is less than or equal to $r+\mu$, which implies that $\val_\a(g')\leq \val_\a(g)+\mu$.
\end{proof}

Let $W=\{\omega_1, \ldots,\omega_n\}$ be a local integral basis at infinity. Let $\lambda \in \set Z$ and $M\in C(x)^{n\times n}$ be such that $W'=x^\lambda MW$ and $\nu_{\infty}(M) = 0$. Then $\lambda = -\nu_{\infty}(x^\lambda M)= \deg_x(x^\lambda M)$. By repeating the reduction at infinity, we can reduce the degree in $x$ as far as possible and decompose $f\in A$ as
\begin{equation}\label{EQ:hr-remainder-infty}
	f= g '  + h\quad \text{with}\quad h =\sum_{i=1}^n h_i\omega_i
\end{equation}
where $g\in A$, $h_i\in C(x)$ with $\deg_x(h_i) < \max\{0,\lambda\}$ for all $i$ and  the coefficients of $g$ are polynomials.
The following lemma derives an upper bound for the degree of any hypothetical integral of $h$ in $A$.
\begin{lemma}\label{LEM:degree}
	Let $h \in A$ be as in ~\eqref{EQ:hr-remainder-infty}. If $h$ is integrable in $A$, then $h=(\sum_{i=1}^n b_i\omega_i)'$ with $b_i\in C(x)$ and $\deg_x(b_i)\leq\max\{0,\lambda\}$ for all $i\in\{1,\ldots,n\}$.
\end{lemma}

\begin{proof}
	%
	Suppose $h$ is integrable in $A$. Then there exists $H=\sum_{i=1}^n b_i\omega_i\in A$ with $b_i\in C(x)$ such that $h=H'$. By~\eqref{EQ:hr-remainder-infty},  we know the coefficients of $h$ satisfy $\deg_x(h_i) < \max\{0,\lambda\}$, which implies $\nu_\infty(h_i)> \min\{0,-\lambda\}$. Since $\{\omega_1,\ldots,\omega_n\}$ is a local integral basis at infinity, it follows from Lemma~\ref{Lem:pole} that $\val_\infty(h)> \min\{0,-\lambda\}$.
	
	We want to show that $\deg_x(b_i)\leq \max\{0,\lambda\}$ for all $i$, which means $\nu_\infty(b_i) \geq \min\{0,-\lambda\}$ for all $i$. Suppose to the contrary that $\tau:=\min_{i=1}^n\{\nu_\a(b_i)\} < \min\{0,-\lambda\}$. Since $\{\omega_1,\ldots, \omega_n\}$ is a local integral basis at infinity, by Lemma~\ref{Lem:pole} we get $\val_\infty(H) < \tau +1 \leq \min\{0,-\lambda\}$. So $H$ has a pole at infinity. By Lemma~\ref{Lem:val_under_der} we have
	\[\val_\infty(h) \leq \val_\infty(H) + 1\leq \min\{0,-\lambda\}.\]
	This leads to a contradiction. So $\deg_x(b_i)\leq \max\{0,\lambda\}$ for all $i$.
\end{proof}

\subsection{The Global Case}\label{sec:global}
To avoid algebraic extensions of the base field, Hermite reduction can be performed simultaneously at all roots of some squarefree polynomial. 

Let  $W=\{\omega_1,\ldots,\omega_n\}$ be an integral basis of $A=C(x)[D]/\<L>$. Let $e\in C[x]$ and $M=(m_{i,j})_{i,j=1}^n\in C[x]^{n\times n}$ be such that $eW'=MW$ and $\gcd(e, m_{1,1},m_{1,2},\ldots, m_{n,n})=1$. Let $v$ be a nontrivial squarefree polynomial and $\lambda\in \set N$ be an integer such that $v^\lambda \mid e$ and $\gcd(\frac{e}{v^\lambda}, v)=1$. Let $f=\frac1{uv^d}\sum_{i=1}^n a_i\omega_i\in A$ with
$u,a_1,\dots,a_n\in C[x]$ such that $d> 1$ and $\gcd(u,v)=\gcd(v,v')=\gcd(v,a_1,\dots,a_n)=1$. Upon differentiating, the $\omega_i$'s may introduce denominators, namely the factors of $e$. Without loss of generality, we assume $e \mid uv^d$. Suppose $d>\max\{1,\lambda\}$. In order to execute one step of the Hermite reduction to reduce the multiplicity $d$, we seek
$b_1,\dots,b_n,c_1,\dots,c_n$ in $C[x]$ such that
\begin{equation}\label{EQ:hr-goal-v}
	\frac{\sum_{i=1}^na_i\omega_i}{uv^d} =\left(\frac{\sum_{i=1}^nb_i\omega_i}{v^{d-1}}\right)^\prime + \frac{\sum_{i=1}^nc_i\omega_i}{uv^{d-1}}.
\end{equation}
If $\lambda=0$, then $\gcd(e, v)= 1$. Multiplying~\eqref{EQ:hr-goal-v} by $uv^{d}$ and reducing this equation modulo $v$ yield
\begin{equation}\label{EQ:hr-case1}
	\sum_{i=1}^n a_i \omega_i \equiv -(d-1)\sum_{i=1}^n b_iuv^\prime \omega_i\mod v.
\end{equation}
Since $\gcd(u,v)=\gcd(v,v^\prime) =1$, we get $b_i \equiv -(d-1)^{-1}(uv^\prime)^{-1} a_i\mod v$ is the unique solution of~\eqref{EQ:hr-case1} in $C[x]/\<v>$.
If $\lambda \geq 1$, multiplying~\eqref{EQ:hr-goal-v} by $uv^{d+\lambda-1}$ and reducing this equation modulo $v^{\lambda-1}$ yield
%
\begin{equation}\label{EQ:hr-v}
	\sum_{i=1}^n v^{\lambda-1} a_i\omega_i \equiv \sum_{i=1}^n b_iuv^{d+\lambda-1}\left(\frac{\omega_i}{v^{d-1}}\right)^\prime \mod v^\lambda.
\end{equation}
One can adapt the argument in the local case to show that~\eqref{EQ:hr-v} always has a unqiue solution~$(b_1, \ldots, b_n)$ in $\left(C[x]/\<v^\lambda>\right)^n$. Let $\psi_i := v^{d+\lambda-1}(v^{1-d}\omega_i)'$ for $i=1,\ldots, n$. As an analog of Proposition~\ref{Prop:psi}, for each root $\alpha\in  \bar C$ of $v$, we have
\begin{equation}\label{EQ:psi_bound}
	\sum_{i=1}^n  \bar C(x)_\a \psi_i \subseteq \OO_\alpha \subseteq \frac{1}{v^{\lambda-1}}\sum_{i=1}^n  \bar C(x)_\a \psi_i.
\end{equation}
Thus the linear system
$\sum_{i=1}^n v^{\lambda-1}a_i \omega_i = \sum_{i=1}^n u b_i\psi_i$
has a unique solution $(b_1, \ldots, b_n)$ in $\left(C[x]/\<v^{\lambda}>\right)^n$.
Equating the coefficients of the $\omega_i$'s on both sides of~\eqref{EQ:hr-v}, the vector $b=(b_1,\ldots, b_n)$ can be found by solving the following linear system of congruence equations:
\begin{equation}\label{EQ:hr-matrix-v}
	(v^{\lambda-1} a_1,\ldots, v^{\lambda -1 }a_n) \equiv b(uv^{\lambda}e^{-1}M -(d-1)uv^{\lambda-1}v'I_n) \mod v^\lambda,
\end{equation}
where $I_n$ is the identity matrix in $C[x]^{n\times n}$.

By repeated application of the above Hermite reduction step, we can reduce the pole orders at finite places as far as possible, i.e.,  we can decompose any $f\in A$ as
\begin{equation}\label{EQ:hr-remainder}
	f =\tilde g' + h\quad\text{with}\quad h=\sum_{i=1}^n\frac{h_i\omega_i}{de},
\end{equation}
where $\tilde g\in A$, $h_1,\ldots,h_n, d\in C[x]$, $d$ is squarefree, $\gcd(d,e) = 1$ and the coefficients of $\tilde g$ are proper rational functions.

\begin{lemma}\label{LEM:d}
	Let $h\in A$ be as in ~\eqref{EQ:hr-remainder}. If $h$ is integrable in $A$, then $h=\left(\frac{\sum_{i=1}^nq_i\omega_i}{u} \right)'$, where $q_1,\ldots,q_n,u\in C[x]$ and $u \mid \gcd(e,e')$. Furthermore, we have $d\in C$.
\end{lemma}
\begin{proof}
	Suppose $h$ is integrable in $A$. Then there exists $H=\sum_{i=1}^n b_i\omega_i\in A$ with $b_i\in C(x)$ such that $h=H'$.
	
	If $\a\in \bar C$ is not a root of $e$, then $b_i$ has no pole at $\a$ for all $i$. Otherwise, suppose $b_i$ has a pole at $\a$ for some $i\in\{1,\ldots,n\}$. Then $H$ has a pole at $\a$, because $\{\omega_1,\ldots,\omega_n\}$ is an integral basis. But then by Lemma~\ref{Lem:val_under_der}, $h$ would have a pole of order greater than $1$, which is impossible because $\gcd(d,e)=1$ and $d$ is squarefree. Therefore $b_i$ is integral at $\a$ for all $i$ and hence $d$ is a constant.
	
	If $\a\in \bar C$ is a root of $e$, then $b_i$ has a pole at $\a$ of order at most $\nu_\a(e)-1$. Otherwise, suppose $\tau:=\min_{i=1}^n\{\nu_\a(b_i)\}\leq -\nu_\a(e)\leq -1$. By Lemma~\ref{Lem:pole}, $\val_\a(H)<\tau+1\leq 0$. Then $H$ has a pole at $\a$. So by Lemma~\ref{Lem:val_under_der}, we have $\val_\a(h)\leq \val_\a(H)-1$. Thus $\val_\a(h) < -\nu_\a(e)$. But from $h=\sum_{i=1}^n\frac{h_i\omega_i}{de}$, we see $\val_\a(h) \geq -\nu_\a(e)$ because $\gcd(d,e)=1$. This leads to a contradiction.
	
	Note that $\nu_\a(\gcd(e,e')) = \nu_\a(e)-1$ if $\a$ is a root of $e$ and $\nu_\a(\gcd(e,e'))=0$ if $\a$ is not a root of $e$. So $u$ is a common multiple of the denominators of the $b_i$'s.
\end{proof}
\begin{example}
	Let $L=x^3 D^2 + (3x^2+2)D\in \set C(x)[D]$ be the same differential operator as in Examples \ref{Ex:non-fuch at 0} and~\ref{Ex:psi}. Since $e=x^3$, we have $u =\gcd(e,e')=x^2$. We could find a Hermite remainder such that the denominator of its integral is $u$, for example,
	\[f=\frac{2\omega_1 + \omega_2}{2x^3} = \left(\frac{-2\omega_1 -\omega_2}{4x^2}\right)'.\]
\end{example}
\section{Additive Decompositions}\label{sec:add}
Now we combine the Hermite reduction at finite places and at infinity to decompose a D-finite function $f$ as $f=g' + h$ such that $f$ is integrable if and only if the remainder $h$ is zero. To achieve this goal, first we confine all remainders into a finite-dimensional vector space and then find all possible integrable functions in this vector space. This procedure is similar to the hyperexponential case~\cite{bostan13a}, the algebraic case~\cite{chen16a}, the Fuchsian case~\cite{chen17a} and the D-finite case~\cite{vanderHoeven21,bostan18a}. It provides an alternative method for solving the accurate integration problem for D-finite functions~\cite{AbramovVanHoeij1999}.

Since there may not exist a basis of $A=C(x)[D]/\<L>$ that is a local integral basis at all $\a\in  \bar C\cup \{\infty\}$, we need two bases to perform Hermite reduction at finite places and at infinity, respectively. Let $W=(\omega_1,\ldots,\omega_n)\in A^n$ be an integral basis of~$A$ that is normal at infinity. There exists $T = \diag\bigl(x^{\tau_1}, \ldots, x^{\tau_n}\bigr) \in C(x)^{n\times n}$ with $\tau_i\in \set Z$ such that $V := TW$ is a local integral basis at infinity. Let $e, a\in C[x]$ and $M, B\in C[x]^{n \times n} $ be such that $eW' = MW$ and $aV'=BV$. Since the derivative of $V$ is
$V' = (TW)' = \biggl(T' + \frac{1}{e}TM\biggr)T^{-1}V,$
we may assume that $a=x^\lambda e$ for some $\lambda\in \set N$.
For any integers $\mu, \delta\in \set Z$ with $\mu \leq \delta$, define a subspace of Laurent polynomials in $C[x,x^{-1}]$ as follows:
\[C[x]_{\mu,\delta}:= \left\{\left.\sum_{i=\mu}^{\delta} a_i x^i\,\right|\,a_i\in C\right\}.\]
\begin{theorem}\label{Thm:add decomp}
	Let $W, V\in A^n$ be as described above. Then any element $f\in A$ can be decomposed into
	\begin{equation}\label{EQ:add}
		f = g' + \frac{1}{d} RW + \frac{1}{x^\lambda e} QV,
	\end{equation}
	where $g\in A$, $d\in C[x]$ is squarefree and $\gcd(d, e)=1$, $R\in C[x]^n$, $Q\in C[x]_{\mu,\delta}^n$ with $\deg_x(R) < \deg_x(d)$, $\mu = \min\{-\tau_1, \ldots, -\tau_n, 0\}$ and $\delta = \max\{\lambda + \deg_x(e),\deg_x(B)\}-1$. Moreover, $f$ is integrable in $A$ if and only if $R=0$ and
	\[\frac{1}{x^\lambda e}QV\in U' \quad \text{with} \quad U =\left\{\left.\frac{1}{u}cV \,\right | \,c\in C[x]_{\mu', \delta'}^n\right\},\]
	where $u =\gcd(e,e')$, $\mu'=\min\{-\tau_1,\ldots, -\tau_n,\nu_0(u)\}$ and \[\delta'=\max\{\deg_x(u), \deg_x(B) - \lambda - \deg_x(e) +\deg_x(u)\}.\]
\end{theorem}

\begin{proof}
	Let~$h\in A$ be a Hermite remainder as in~\eqref{EQ:hr-remainder}. By the extended Euclidean algorithm, we compute $r_i, s_i\in C[x]$ such that
	$h_i = r_i e + s_i d$ and $\deg_x(r_i) < \deg_x(d)$. Then~$h$
	decomposes as
	\begin{equation}\label{EQ:h-remainder}
		h=\sum_{i=1}^n \frac{h_i}{de}\omega_i =
		\sum_{i=1}^n \frac{r_i}{d}\omega_i +
		\sum_{i=1}^n \frac{s_i}{e}\omega_i.
	\end{equation}
	Writing $h$ in vector form, by~\eqref{EQ:hr-remainder} we decompose $f\in A$ as
	\begin{equation}\label{EQ:add1}
		f = \tilde{g}' + \frac{1}{d} RW + \frac{1}{e} SW,
	\end{equation}
	where $\tilde g\in A$, $R = (r_1, \ldots, r_n)\in C[x]^n$, $S = (s_1, \ldots, s_n)\in C[x]^n$. In the next step, we shall reduce the degree of $S$ and confine $S$ to a finite-dimensional vector space over $C$ that is independent of $f$.
	We rewrite the last summand in~\eqref{EQ:add1} with respect to the new basis~$V$:
	\begin{equation}\label{EQ:W-V}
		\frac{1}{e} SW = \frac{1}{x^\lambda e} \tilde{S}V,
	\end{equation}
	where $\tilde{S} = x^\lambda S T^{-1} \in x^\mu C[x]^n$ with $\mu =\min\{-\tau_1, \ldots, -\tau_n,0\}$. Since $V$ is a local integral basis at infinity, using Hermite reduction at infinity in Section~\ref{sec:local}, we obtain from~\eqref{EQ:hr-remainder-infty} that
	\begin{equation}\label{EQ:S1S2}
		\frac{1}{e} SW = (S_1 V)' + \frac{1}{x^\lambda e} S_2 V,
	\end{equation}
	where $S_1\in C[x]^n$ and $S_2\in x^\mu C[x]^n$ satisfies
	\[\deg_x\left(\tfrac{S_2}{x^\lambda e}\right) \leq \max\left\{0, \deg_x\left(\tfrac{B}{x^\lambda e}\right)\right\}-1.\] This implies that $\deg_x(S_2)\leq \max\{\lambda + \deg_x(e),\deg_x(B)\}-1=\delta$. Thus $S_2\in C[x]_{\mu,\delta}^n$ and we finally obtain the decomposition~\eqref{EQ:add} by setting	$g = \tilde{g} + S_1 V$ and $Q = S_2$.	
	
	For the last assertion, assume that $f$ is integrable (the other direction of
	the equivalence holds trivially). Then Lemma~\ref{LEM:d} implies that $d\in C$,
	and therefore $R$ must be zero because $\deg_x(R) < \deg_x(d)$. Hence the last
	summand in~\eqref{EQ:add} and the left hand side of~\eqref{EQ:S1S2} are also integrable. We want to find its integral by estimating the valuation of this integral at all points in $ C\cup \{\infty\}$. Since $W$ is a global integral basis, using Lemma~\ref{LEM:d} again, we know
	\[
	\frac{1}{e} SW = \left(\frac{1}{u}bW\right)',
	\]
	where $b\in C[x]^n$ and $u =\gcd(e,e')$. Then
	\[\frac{1}{x^\lambda e} QV= \frac{1}{e}SW - (S_1 V)' = \left(\left(\frac{bT^{-1}}{u} - S_1\right)V\right)'=\left(\frac{1}{u}cV\right)',\]
	where $c = bT^{-1} - uS_1\in C[x,x^{-1}]^n$. Now we only need to estimate the valuation of $c$ at the remaining two points $0$ and $\infty$. By the expression of $c$, we get
	\[\nu_0(c) \geq\min\{\nu_0(bT^{-1}), \nu_0(u S_1)\} \geq \min\{-\tau_1, \ldots, -\tau_n, \nu_0(u)\}=\mu'.\]
	On the other hand, since $V$ is a local integral basis at infinity, it follows from Lemma~\ref{LEM:degree} that $\deg_x\left(\frac{c}{u}\right) \leq \max\left\{0, \deg_x\left(\frac{B}{x^\lambda e}\right)\right\}$. Therefore
	\[\deg_x(c)\leq \max\{\deg_x(u), \deg_x(B) - \lambda - \deg_x(e) +\deg_x(u)\}=\delta'.\]
	Finally we have $c\in C[x]_{\mu',\delta'}^n$.
\end{proof}

The remaining step is to reduce all integrable D-finite functions to zero. Note that in Theorem~\ref{Thm:add decomp}, $U$ is a $C$-vector space of dimension $n(\delta'-\mu'+1)$ with a basis
\begin{equation}
	\left\{\left.\frac{x^j v_i}{u} \,\right|\, i=1,\ldots, n; j=\mu', \ldots, \delta'\right\},
\end{equation}
where $V=(v_1, \ldots, v_n)$.
Let $K=\left\{\left.\frac{1}{x^\lambda e} b V \, \right|\, b\in C[x]_{\mu,\delta}^n\right\}$. Differentiating all elements in the basis of $U$ and using Gauss elimination, we can find a basis of $U'$ and decompose $K= (U'\cap K) \oplus N_V$ as a direct sum of two subspaces, where $N_V$ is a complement of $U'\cap K$ in $K$. This means $f$ in~\eqref{EQ:add} can be further decomposed as
\begin{equation}\label{EQ:add2}
	f = \tilde g' + \frac{1}{d} RW + \frac{1}{x^\lambda e} Q_2V,
\end{equation}
where $\tilde g= g+ g_1$ with $g_1' \in U'\cap K$ and $Q_2\in C[x]_{\mu,\delta}^n$ such that $f$ is integrable in $A$ if and only if $R=Q_2=0$. This decomposition~\eqref{EQ:add} is called an \emph{additive decomposition} of~$f$ with respect to~$x$. When $L$ is a Fuchsian operator, the additive decomposition of $f$ was obtained in~\cite[Theorem 23]{chen17a}.

In practice, we may choose a fixed complement of $K\cap U'$ in $K$. To do this, we define a {\em term over position order} on the set
\[\left\{\left.x^jv_i \,\right| \, i=1,\ldots, n; j\in \set Z\right\}\]
such that $x^{j_1}v_{i_1}>x^{j_2}v_{i_2}$ if and only if $j_1 >j_2$ or $j_1=j_2$ and $i_1<i_2$. Let $\lt(\cdot)$ denote the leading term of an element in $A=C(x)[D]/\<L>$. For example, if $p=3x^2(v_1+v_2)+10xv_1\in A$, then $\lt(p)=x^2v_1$. Then a {\em standard complement} of $K\cap U'$ in $K$ is a $C$-vector subspace of $K$ generated by
\[\{ h \in K \mid h \neq \lt(g) \text{ for all } g\in K \cap U' \}.\]
From now on, let $N_V$ denote the standard complement of $K\cap U'$ in $K$.  This definition of $N_V$ is essentially the same as in~\cite{chen17a}, because there is a bijection from a D-finite function to its coefficients of the $v_i$'s.

Note that $Q_2$ belongs to a $C$-vector space $C[x]_{\mu,\delta}^n$ of dimension
\begin{equation}\label{EQ: bound_Nv}	n(\delta-\mu+1)=\max\{\lambda+\deg_x(e),\deg_x(B)\}+\max\{\tau,0\},
\end{equation}
where $\tau=\max\{\tau_1,\ldots, \tau_n\}$. If $L$ is Fuchsian, by~\cite[Lemma 4]{chen17a}, we know $\deg_x(B)<\lambda + \deg_x(e)$. So $Q_2$ belongs to a $C$-vector space of dimension at most $n(\max\{\tau,0\}+\lambda+\deg_x(e))$.
This is a refinement of~\cite[Proposition 22]{chen17a}.

\begin{example}\label{Ex:add-nonfuch-infty}
	Let $L=x D^2 - (3x^3+2)D\in \set C(x)[D]$ be the same operator as in Example~\ref{Ex:non-fuch at infty}. Then $W=(\omega_1,\omega_2)=(1,x^{-2}D) = V$. So $e = 1$, $\lambda =0$ and $M=B=\binom{0\hspace{0.3cm}x^2}{0\,~\,3x^2}$. After performing Hermite reduction at infinity in Example~\ref{Ex:hr-infty}, we get
	\begin{equation}\label{EQ:add-nonfuch}
		f=\left(x^4\omega_1+\left(\tfrac{4}{9}x-\tfrac{1}{3}x^4\right)\omega_2\right)'+\left(x -\tfrac{4}{9}\right)\omega_2.
	\end{equation}
	Then $\mu =0$, $\delta=1$, $u=1$, $\mu'= 0$, $\delta'=2$. A basis of $U$ is
	\[\{\omega_1, \,\omega_2, \, x\omega_1, \, x\omega_2,\, x^2\omega_1,\, x^2\omega_2\},\]
	and hence $U'$ is generated by
	\[\{x^2\omega_2, \, 3x^2\omega_2, \, \omega_1+x^3\omega_2, \, (1+3x^3)\omega_2, \, 2x\omega_1+x^4\omega_2,\, (2x+3x^4)\omega_2\}.\]
	So a basis of $K \cap U'$ is
	\[\left\{3\omega_1-\omega_2,\, 6x\omega_1-2x\omega_2\right\}.\]
	The leading terms of all elements in $K\cap U'$ are $\omega_1$ or $x\omega_1$.
	Since $\lt((x-\frac{4}{9})\omega_2) = x\omega_2$ is different from all these terms, by Theorem~\ref{Thm:add decomp} we know $f$ is not integrable in $A=\set C(x)[D]/\<L>$ and~\eqref{EQ:add-nonfuch} is an additive decomposition of $f$ with respect to $x$.
\end{example}

\begin{example}\label{Ex:add-fuch-infty}
	Let $L=x^3 D^2 + (3x^2+2)D\in \set C(x)[D]$ be the same operator as in Example~\ref{Ex:non-fuch at 0}. Then $W=(\omega_1,\omega_2)=(1,x^3D) = V$. So $e = x^3$, $\lambda =0$ and $M=B=\binom{0\hspace{0.3cm}1}{0\,-2}$. Combining Hermite reduction at all finite places in Example~\ref{Ex:hr-finite} and Hermite reduction at infinity, we get
	\[f = \left(\left(\tfrac{2}{3x}-x\right)\omega_1+\left(\tfrac{4}{3x}-3x\right)\omega_2\right)'-\tfrac{4}{3x^2}\omega_1-\tfrac{2}{3x^2}\omega_2.\]
	Then $\mu =0$, $\delta=2$, $u=x^2$, $\mu'= 0$, $\delta'=2$. A basis of~$U$ is
	\[\left\{\tfrac{\omega_1}{x^2},\,\tfrac{\omega_2}{x^2},\,\tfrac{\omega_1}{x},\,\tfrac{\omega_2}{x},\,\omega_1, \,\omega_2\right\}.\]
	A basis of $K \cap U'$ is
	\[\left\{-\tfrac{2}{x^2}\omega_1-\tfrac{1}{x^2}\omega_2,\, \tfrac{1}{x^3}\omega_2,\, -\tfrac{2}{x^3}\omega_2\right\}.\]
	Therefore $f$ is integrable:
	\[f = \left(\left(\tfrac{2}{3x}-x+\tfrac{4}{3x}\right)\omega_1+\left(\tfrac{4}{3x}-3x+\tfrac{2}{3x}\right)\omega_2\right)'.\]
\end{example}

\section{Applications}\label{sec:telescoper}
Let $K(x)[\partial_t, D_x]$ with $K=C(t)$ be an Ore algebra, in which $D_x$ is the differentiation with respect to $x$ and $\partial_t$ is either the differentiation with respect to $t$ or the shift $t \mapsto t+1$. Let $I$ be a left ideal of $K(x)[\partial_t, D_x]$ generated by $L$ and $\partial_t - u_t$ with $L, u_t\in K(x)[D_x]$. The quotient $A=K(x)[\partial_t, D_x]/I$ is a finite-dimensional vector space over $K(x)$ and a basis is given by $\{1, D_x, \ldots, D_x^{n-1}\}$, where $n$ is the order of $L$. Every element $f$ in $A$ can be uniquely written as $P_f+I$ for some $P_f \in K(x)[D_x]$. The map sending $f$ to $P_f+\<L>$ gives an isomorphism from $A$ to $K(x)[D_x]/\<L>$ as a $K(x)[D_x]$-module. Using this isomorphism, for any $f\in A$, we can apply our additive decomposition to test whether $f$ is integrable (in $x$). If $f\in A$ is not integrable, one can ask to find a nonzero operator $T \in C(t)[\partial_t]$ (free of $x$) such that $T(f)$ is integrable. Such an operator $T$ if it exists is called a \emph{telescoper} for $f$. Applying the additive decomposition with respect to $x$ in Section~\ref{sec:add} to $\partial_t^i\cdot f\in A$ yields that
\[\partial_t^i\cdot f = g_i'+h_i\]
where $g_i,h_i\in A$, and $\partial_t^i\cdot f$ is integrable in $A$ if and only if $h_i=0$. If there exist $c_0, c_1,\ldots, c_{r}\in K$ such that $\sum_{i=0}^r c_i h_i=0$, then $T=\sum_{i=0}^rc_i\partial_t^i$ is a telescoper for $f$ (because $\partial_t$ and $D_x$ commute). Such a telescoper if it exists is of minimal order. This approach is the method of reduction-based telescoping and was developed for various classes of functions~\cite{bostan10b,bostan13a,chen15a,chen16a, chen17a, bostan18a, vanderHoeven21}. If $\partial_t=D_t$ is the differentiation with respect to $t$, then telescopers always exist~\cite{Zeilberger1990}. We implemented our algorithm in Maple. More examples and our code are available in~\cite{website}. Similar to the Fuchsian case~\cite[Lemma 24]{chen17a}, for any $i\in\set N$, the derivative $D_t^i \cdot f$ has an additive decomposition~\eqref{EQ:add2} of the form
\begin{equation}
	D_t^i\cdot f = g_i'+h_i\quad\text{with}\quad h_i = \frac{1}{d}R_iW + \frac{1}{x^\lambda e}Q_iV
\end{equation}
where $g_i\in A$, $d\in K[x]$, $R_i \in K[x]^n$, $Q_i \in K[x, x^{-1}]^n$,  with $\deg_x(R_i) < \deg_x(d)$ and $Q_i\in N_V$. Then by~\eqref{EQ: bound_Nv} we obtain an upper bound for the order of telescopers, which is a generalization of~\cite[Corollary 25]{chen17a}.


\begin{corollary}
	Every $f\in A$ has a telescoper of order at most $n\deg_x(d) + \dim_x(N_V)$, which is bounded by \[n(\deg_x(d) + \max\{\tau,0\}+\max\{\lambda+\deg_x(e),\deg_x(B)\} ).\]
\end{corollary}
\begin{example}
	Let $H =\sqrt{t-2x}\exp(t^2x)$ be a hyperexponential function. This function is annihilated by
	\[L =D_x - \frac{2 t^2 x - t^3 +1}{2x-t} \quad \text{and}\quad  D_t - \frac{8 t x^2 - 4 t^2 x - 1}{2(2x-t)}.\]
	An integral basis of $A=K(x)[D_x]/\<L>$ with $K=\set C(t)$ is $\omega = 1$ and a local integral basis at infinity is $v =x^{-1}\omega$. 
	As the integrand $H$ corresponds to $1\in A$, its representation in the bases is $f = \omega = xv$. The additive decomposition of $f$ is
	\[f = \left(\left(\frac{x}{t^2}-\frac{1}{2t^2}\right)v\right)' -\frac{(t^3+1)x - t}{2t^4x(2x-t)}v.\]
	Next we consider the derivative $D_t \cdot f$ which has an additive decomposition
	\[D_t\cdot f = \left(\left(\frac{2x^2}{t} - \frac{3x}{t^3} - \frac{3t^3 - 6}{4t^5}\right)v\right)' - \frac{3(t^3 - 2)((t^3+1)x - t )}{4t^5x(2x-t)}v.\]
	Now we see the reminders of $f$ and $D_t \cdot f$ are linearly dependent over $\set C(t)$, which gives rise to a telescoper $2tD_t-3 (t^3-2)$. This telescoper was obtained in~\cite[Example 21]{bostan13a} with a different reduction approach.
\end{example}
\begin{example}
	Let $F_n(x)=x^nJ_n(x)$ where $J_n$ denotes the Bessel function of the first kind. Then $F_n(x)$ is annihilated by
	\[L =D_x^2 + (1 - 2n)D_x + x\quad\text{and}\quad P =  S_n + xD_x - 2n,
	\]
	where $S_n$ is the shift operator with respect to $n$. An integral basis of $A=K(x)[D_x]/\<L>$ with $K= \set C(n)$ is $W =(\omega_1, \omega_2) = (1, D_x)$ and a local integral basis at infinity is $V =(v_1, v_2) =  (\omega_1, x^{-1}\omega_2)$. As before, $F_n(x)$ is represented by $f= 1 \in A$. The additive decompositions of $f$ and $S_n\cdot f = -xD_x + 2n \in A$ are as follows:
	\[f =(v_2)' + \frac{(2n - 1)x-1}{x} v_2,\]
	\[S_n\cdot f = (-xv_1 -(2n +1)v_2)' +\frac{(2n + 1)((2n -1) x - 1)} {x} v_2.\]
	Now we can find a telescoper $S_n - 2n - 1$. This was obtained by the algorithm in~\cite{bostan18a}.
\end{example}
\noindent\textbf{Acknowledgement.}
We thank Shayea Aldossari for sharing his Maple package {\em integral\_bases}.

\bibliographystyle{plain}
\bibliography{integral}

\end{document}